\documentclass[conference]{IEEEtran}
\IEEEoverridecommandlockouts

\usepackage{cite}
\usepackage{amsmath,amsthm,amssymb,graphicx,algorithm,algorithmic}
\usepackage{color}
\usepackage{makecell}
\usepackage{subfigure}
\newtheorem{theorem}{\bf Theorem}

\usepackage{textcomp}
\def\BibTeX{{\rm B\kern-.05em{\sc i\kern-.025em b}\kern-.08em
    T\kern-.1667em\lower.7ex\hbox{E}\kern-.125emX}}
\begin{document}

\title{Coverage Probability Analysis of RIS-Assisted High-Speed Train Communications}
\DeclareRobustCommand*{\IEEEauthorrefmark}[1]{%
  \raisebox{0pt}[0pt][0pt]{\textsuperscript{\footnotesize\ensuremath{#1}}}}
    \author{\IEEEauthorblockN{Changzhu Liu\IEEEauthorrefmark{1}, Ruisi He\IEEEauthorrefmark{1*}, Yong Niu\IEEEauthorrefmark{1*}, Bo Ai\IEEEauthorrefmark{1,2}, Zhu Han\IEEEauthorrefmark{3}, Zhangfeng Ma\IEEEauthorrefmark{1}, \\ Meilin Gao\IEEEauthorrefmark{4}, Zhangdui Zhong\IEEEauthorrefmark{1}, Ning Wang\IEEEauthorrefmark{5}
}

\IEEEauthorblockA{\IEEEauthorrefmark{1}State Key Laboratory of Rail Traffic Control and Safety, Beijing Jiaotong University, Beijing, China}
\IEEEauthorblockA{\IEEEauthorrefmark{2}Research Center of Networks and Communications, Peng Cheng Laboratory, Shenzhen, China}
\IEEEauthorblockA{\IEEEauthorrefmark{3}Department of Electrical and Computer Engineering, University of Houston, Houston, TX 77004 USA}
\IEEEauthorblockA{\IEEEauthorrefmark{4}Beijing National Research Center for Information Science and Technology, Tsinghua University, China}
\IEEEauthorblockA{\IEEEauthorrefmark{5}School of Information Engineering, Zhengzhou
University, Zhengzhou 450001, China}


*Corresponding Author: Ruisi He, Yong Niu (ruisi.he@bjtu.edu.cn;niuyong@bjtu.edu.cn)

}



\maketitle

\begin{abstract}
Reconfigurable intelligent surface (RIS) has received increasing attention due to its capability of extending cell coverage by reflecting signals toward receivers. This paper considers a RIS-assisted high-speed train (HST) communication system to improve the coverage probability. We derive the closed-form expression of coverage probability. Moreover, we analyze impacts of some key system parameters, including transmission power, signal-to-noise ratio threshold, and horizontal distance between base station and RIS. Simulation results verify the efficiency of RIS-assisted HST communications in terms of coverage probability.

\end{abstract}


\section{Introduction}
High-speed train (HST) communications have attracted a lot of attention in recent years. Compared with the traditional wireless communication networks, HST communication has high mobility of onboard transceivers and large signal penetration loss through train cars, and it leads to many challenges such as channel modeling, coverage enhancement, Doppler shift compensation, time-varying channel estimation, beamforming design, and resource management\cite{t1,He,Huang1,Huang2}. Reconfigurable intelligent surface (RIS) is one of key technologies for future wireless communication, and it is capable of smartly designing radio environment \cite{t2,sun}. Moreover, the wireless coverage can be expected to increase with the aid of RIS \cite{t3}, and it is thus helpful for HST communication enhancement.

In the design of HST communication system, cell coverage performance is an important indicator. However, in terms of coverage performance analysis and improvement of HST communications, there are only few works in the literature. In \cite{t6}, a space-ground integrated cloud railway network was proposed in order to achieve seamless coverage for environment-diverse HST, which can reduce handover times and extend coverage of railway. In \cite{t7} and \cite{t8}, a beamforming based coverage performance improvement scheme was proposed for HST communication systems with elliptical cells. In \cite{t9} and \cite{t10}, the authors considered overlap area between adjacent cells and hard handoff scheme, and coverage performance for HST system was further analyzed. In \cite{t11}, coverage efficiency of HST communications was improved by exploiting the radio-over-fiber technology. In \cite{t12}, two different free-space-optics coverage models was proposed in order reduce the impact of handoff processes. In \cite{t13} and \cite{t14}, the authors considered unmanned aerial vehicles assisted HST communications for coverage improvement. The coverage analysis of HST communication system with carrier aggregation was investigated in \cite{t15}, where theoretical expressions for edge coverage probability and percentage of cell coverage area were derived.

RIS-assisted HST communications have been rarely investigated. In \cite{t1}, the authors provided a novel RIS-aided HST wireless communication paradigm, including its main challenges and application scenarios, and provided effective solution to solved the problem of signal processing and resource management. In \cite{t16}, outage probability of a RIS-assisted multiple-input-multiple output (MIMO) downlink system with statistic channel state information (CSI) for HST communications was investigated. In \cite{t17}, the authors considered a train-ground time division duplexing wireless mobile communication paradigm to deploy two RISs for HST communication system, and further solved the spectrum effective maximization problem. In \cite{t18}, the authors investigated spectral efficiency of a RIS-assisted millimeter-wave HST communication by exploiting the deep reinforcement learning method. The RIS-assisted free-space-optics communications for HST access connectivity were investigated in \cite{t19}, and the average signal-to-noise ratio (SNR) and the outage probability were analyzed. In \cite{tt19}, the authors investigated the interference suppression for an RIS-assisted railway wireless communication system to maximize signal-to-interference-plus-noise ratio.

To the best of the authors' knowledge, there is few literature dealing with wireless coverage probability analysis for RIS-assisted HST communications. Motivated by the above gap, this paper investigates coverage performance in downlink single-input-single-output (SISO) RIS-assisted HST communications, as shown in F{}ig.~\ref{fig:system}, where a single-antenna base station (BS) serves a single-antenna mobile relay (MR) with the help of one RIS. We consider the practical case where RIS only has a finite number of discrete phases. Considering the complexity and validity, we exploit the local search method to optimize the RIS phase. Moreover, we analyze the impact of critical system parameters, including transmission power, SNR threshold, and horizontal distance between BS and RIS, on coverage probability.The major contributions of our work can be summarized as follows: 

\begin{itemize}
\item We consider a RIS-assisted SISO downlink HST communication model to extended the coverage of HST, where RIS is deployed within the coverage of BS to provide reflective paths to enhance the received power at MR. The RISs communicate with BS and the MR mounted on top of the train.
\item The closed-form expression of coverage probability is derived for RIS-assisted HST communications. Then, RIS phase is optimized by exploiting local search method.
\item The simulation demonstrates the impact of some key system parameters on coverage probability including transmission power, SNR threshold, and horizontal distance between BS and RIS. The simulation results for the traditional HST communication systems without using RIS is provided for the sake of comparison.
The results demonstrate that the deployment of RISs can effectively improve coverage probability and enhance coverage of HST communication systems.
\end{itemize}

The rest of this paper is organized as follows. The system model is introduced in Section II. In Section III, the coverage probability is derived and RIS phases are optimized. In Section IV, numerical results are presented to show the impact of key system parameters on coverage probability. Section V concludes this work.

\section{System model}
\subsection{Scenario Description}
In this section, we introduce a RIS-assisted HST communication system model, as shown in F{}ig.~\ref{fig:system}, where a single-antenna BS serves a single-antenna MR with the help of RIS. The trackside BS transmits signals to a train through an MR mounted on top of the train avoiding penetration loss. RIS with $N$ reflecting elements is deployed along the rail track. We assume that the BS and MR are in far field of the RIS and each element is capable of independently rescattering signal, which can be dynamically adjusted by RIS controller\cite{t20}, and assume that CSI of all channels is perfectly known \cite{t21}.

We consider total time $T$ slots, $\tau$ is the slot duration, and $\mathcal{N} = \left\{1,\cdots, N \right\}$. $x(t)\sim \mathcal{C} \mathcal{N} (0,1)$ denotes the signal transmitted to MR during time slot $t \in \left\{1,\cdots, T \right\}$ with zero mean and variance equal $1$. The received signal at MR in time slot $t$ can given as
\begin{align}
  y\left( t \right) =\sqrt{P}\left( \underbrace{ h_{\rm{d}} \left( t \right) + \sum_{n=1}^N{h^{n}_{{\rm{r}} }\left( t \right) e^{j\theta _n\left( t \right)}g_{n}}\left( t \right)}_{h\left( t \right)}   \right) x\left( t \right)+ z\left( t \right),
\end{align}
where $P$ denotes transmission power, $h\left( t \right)$ denotes the equivalent channel between BS and MR. The channels in time slot $t$ from BS to MR, from BS to the $n$-th RIS, and from the $n$-th RIS to MR, are denoted by $h^{n}_{{\rm{d}}}\left( t \right)\in \mathbb{C}$, $g_{n}\left( t \right) \in \mathbb{C}$ and $h_{\rm{r}}^{n}\left( t \right) \in \mathbb{C}$, respectively, $z\left( t \right) \sim \mathcal{C} \mathcal{N} (0,\sigma ^2)$ denotes i.i.d. additive white Gaussian noise at MR. $\theta _n\left( t \right)$ represents phase of the $n$-th RIS element. For ease of actual implementation, we consider that the phase at each element of RIS can only take a finite discrete values with equal quantization intervals $\left[0, 2\pi \right)$. Let $b$ denote the number of the quantization bits. Then the set of phases at each element is given by $ \varTheta  =\left\{ 0, \Delta\theta ,\cdots ,\Delta\theta  \left( M-1 \right) \right\} $, where $ \Delta\theta  =\frac{2\pi}{M}$ and $M=2^b$. Note that, the transceiver distances (BS-MR and RIS-MR links) always change across time slots due to high mobility of HST. 

Accordingly, the SNR at MR in time slot $t$ is given by
\begin{equation}
  \gamma \left( t \right) = \frac{P\left|h_{\rm{d}} \left( t \right) + \sum_{n=1}^N{h^{r}_{{\rm{n}} }\left( t \right) e^{j\theta _n\left( t \right)}g_{n}}\left( t \right)  \right|^2 }{\sigma^2},
\end{equation}
where $\left| \cdot \right|$ is the absolute value of a complex number.

\begin{figure}[!t]
  \centering  
    \includegraphics[scale=0.075]{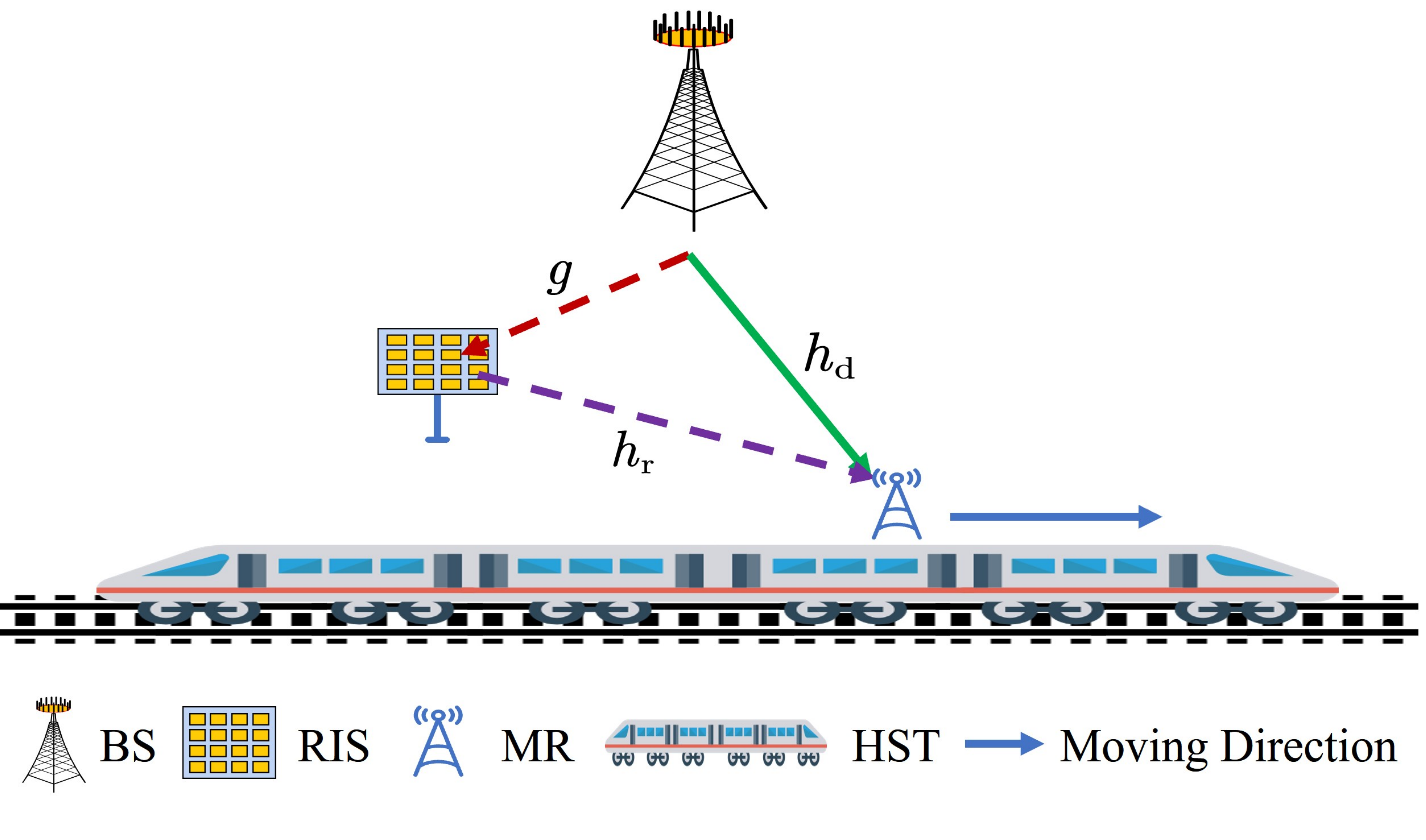}\\
    \caption{\label{fig:system} A RIS-assisted SISO downlink system for HST communications.}
\end{figure}

\subsection{Channel Model}
We assume that all the links follow the Rician fading since line-of-sight (LoS) and non-line-of-sight (NLoS) components exist \cite{t22}. In time slot $t$, BS-RIS link $h_{\rm{d}} \left( t \right)$ can be expressed as 
\begin{equation} \label{eq:hbm}
  h_{\rm{d}} \left( t \right)=\sqrt{\frac{\kappa _{\rm{d}}}{\kappa _{ \rm{d} }+1}}\bar{h}_{ \rm{d}}\left( t \right)+\sqrt{\frac{1}{\kappa _{ \rm{d} }+1}}\tilde{h}_{ \rm{d}}\left( t \right),
\end{equation}
where $\kappa_{\rm{d}} \geq 0 $ is the Rician K-factor,  $\bar{h}_{\rm{d}}\left( t \right)\in \mathbb{C} $ is the LoS component depending on BS-MR link and remains stable with each time slot\footnote{We ignore shadowing effects and assume that large-scale component is determined only by distance-based path-loss \cite{t23}.}, $\tilde{h}_{ \rm{d}}\left( t \right) \in \mathbb{C} $ is the NLoS component. The LoS component of the channel between BS and  MR can be given as \cite{t24}  $\bar{h}_{\rm{d}}\left( t \right)=\sqrt{D_{\rm{d}}^{-\varepsilon_{\rm{d}}}\left( t \right)}e^{-j\theta_{\rm{d}}\left( t \right)}$, where $D_{\rm{d}}\left( t \right)$ denotes distance between BS and MR, as shown in F{}ig.~\ref{fig:top}, $\varepsilon_{\rm{d}}\left( t \right)$ denotes path-loss parameter, and $\theta_{\rm{d}}\left( t \right)$ denote phase. Similarly, the NLoS component can be written as $\tilde{h}_{ \rm{d}}\left( t \right) = \sqrt{d_{\rm{NLoS,d}}^{-\varepsilon'_{\rm{d}}}\left( t \right)}\tilde{h}_{\rm{NLoS}}^{\rm{d}}\left( t \right)$, where $d_{\rm{NLoS,d}}\left( t \right)$ is the distance between BS and MR for the NLoS case, $\varepsilon'_{\rm{d}}$ is the path-loss parameter, $\tilde{h}_{\rm{NLoS}}^{\rm{d}}\left( t \right)\sim \mathcal{C} \mathcal{N} (0,1)$ denotes small-scale channel component. 

\begin{figure}[!t]
  \centering  
  \includegraphics[scale=0.16]{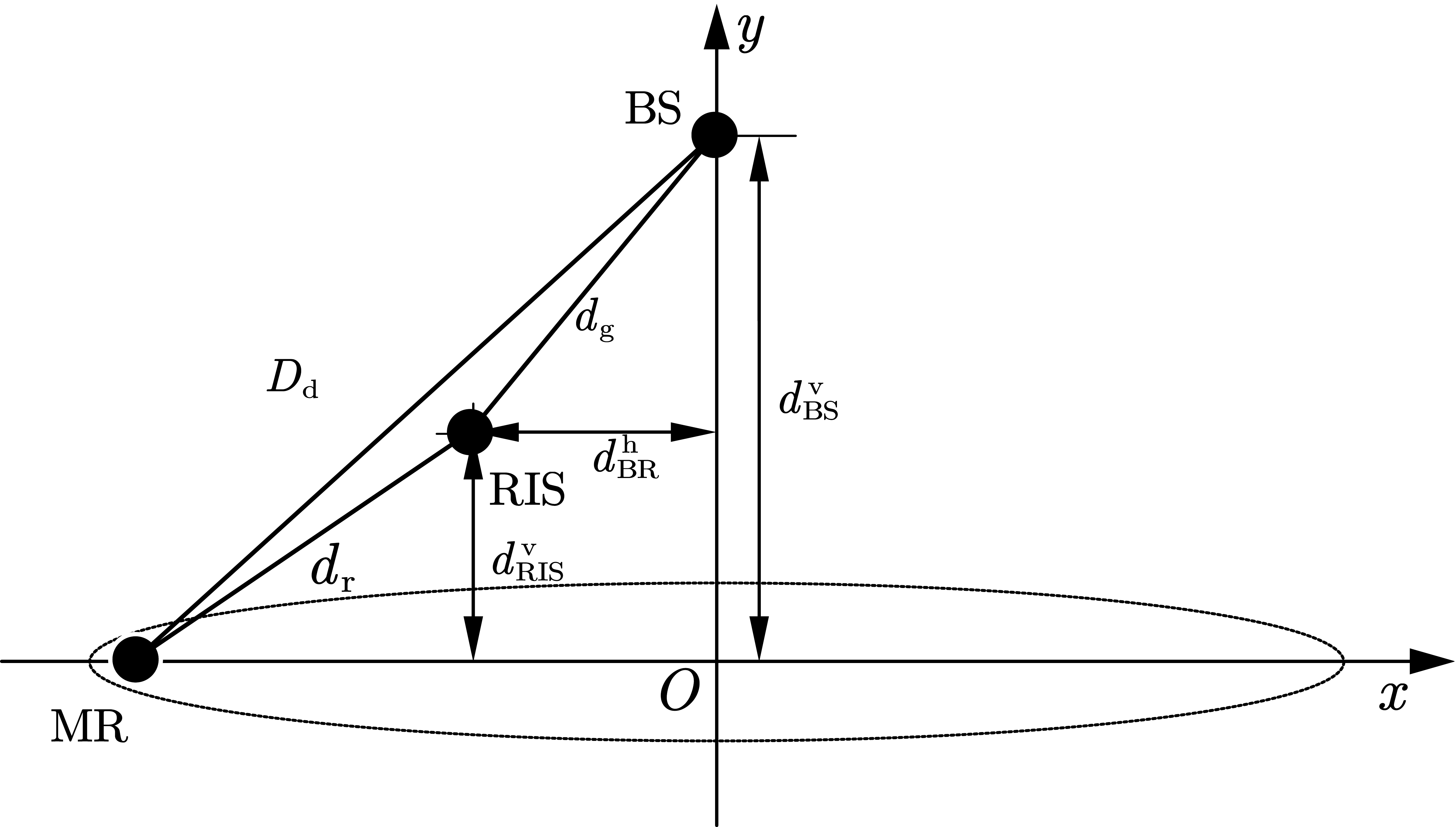}
  \caption{\label{fig:top}Top view of RIS-assisted HST communication system model.}
\end{figure}

For all $n \in \mathcal{N}$, in time slot $t$, $g_{n}\left( t \right)$ and $h_{\rm{r}}^{n}\left( t \right)$ can be given as
\begin{equation} \label{eq:hbr}
  g_{n} \left( t \right)=\sqrt{\frac{\kappa _{\rm{g}}}{\kappa _{ \rm{g}}+1}}\bar{h}_{\rm{g}}^{n}+\sqrt{\frac{1}{\kappa _{ \rm{g}}+1}}\tilde{h}_{\rm{g}}^{n}\left( t \right),
\end{equation}
\begin{equation} \label{eq:hrm}
  h_{\rm{r}}^{n} \left( t \right)=\sqrt{\frac{\kappa _{\rm{r}}}{\kappa _{ \rm{r}}+1}}\bar{h}_{ \rm{r}}^{n}\left( t \right)+\sqrt{\frac{1}{\kappa _{ \rm{r} }+1}}\tilde{h}_{ \rm{r}}^{n}\left( t \right),
\end{equation}
where $\kappa_{\rm{g}} \geq 0 $ and $\kappa_{\rm{r}} \geq 0 $  are the Rician K-factors, $\bar{h}_{ \rm{g}}^{n}\in \mathbb{C}$ and  $ \bar{h}_{\rm{r}}^{n}\left( t \right) \in \mathbb{C}$ denote LoS components, and $\tilde{h}_{ \rm{g}}^{n}\left( t \right)\in \mathbb{C}$ and $\tilde{h}_{\rm{r}}^{n}\left( t \right) \in \mathbb{C}$ denote NLoS components. Similar to BS-MR link, we have $\bar{h}_{ \rm{g}}^{n} =  \sqrt{\left(d_{\rm{g}}^{n}\right )^{-\varepsilon_{\rm{g}}}}e^{-j\theta_{\rm{g}}^{n}}$, $\bar{h}_{\rm{r}}^{n}\left( t \right) =  \sqrt{\left(d_{\rm{r}}^{n} \left( t \right)\right )^{-\varepsilon_{\rm{r}}}}e^{-j\theta_{\rm{r}}^{n}\left( t \right)}$, where $d_{\rm{g}}^{n}$ and $d_{\rm{r}}^{n}\left( t \right)$ denote distance between BS and the $n$-th RIS and between the $n$-th RIS element and MR, as shown in F{}ig.~\ref{fig:top}, respectively. Moreover, $d_{\rm{BS}}^{\rm{v}}$ denotes vertical distance from BS to rail track, $d_{\rm{RIS}}^{\rm{v}}$ denotes vertical  distance from RIS to rail track, and $d_{\rm{BR}}^{\rm{h}}$ denotes horizontal distance between BS and RIS. Parameters $\varepsilon_{\rm{g}}$ and $\varepsilon_{\rm{r}}$ are path-loss parameters, and $\theta_{\rm{g}}^{n}$ and $\theta_{\rm{r}}^{n}\left( t \right)$ are phase. Similarly, the NLoS component can be written as  $\tilde{h}_{ \rm{g}}^{n}\left( t \right) = \sqrt{\left(d_{\rm{NLoS,g}}^{n}\right) ^{-\varepsilon'_{\rm{g}}}}\tilde{h}_{\rm{NLoS,g}}^{n}\left( t \right)$, $\tilde{h}_{ \rm{r}}^{n}\left( t \right) = \sqrt{\left(d_{\rm{NLoS,r}}^{n}\left( t \right)\right)^{-\varepsilon'_{\rm{r}}}}\tilde{h}_{\rm{NLoS,r}}^{n}\left( t \right)$, where $d_{\rm{NLoS,g}}^{n}$, $d_{\rm{NLoS,r}}^{n}\left( t \right)$ denote distance between BS and the $n$-th RIS and between the $n$-th RIS and  MR in the NLoS case, respectively. Parameters $\varepsilon'_{\rm{g}}$, $\varepsilon'_{\rm{r}}$ are the path-loss parameters in the NLoS case, and $\tilde{h}_{\rm{NLoS,g}}^{n}\left( t \right), \tilde{h}_{\rm{NLoS,r}}^{n}\left( t \right)\sim \mathcal{C} \mathcal{N} (0,1)$ denote the small-scale components. 

\section{Coverage Probability Analysis and RIS Phase Optimization}
In this section, we derive the expression of coverage probability for RIS-assisted HST communications. Then, we proposed an local search method to optimize the RIS phase.

\subsection{Coverage Probability}
Coverage probability is defined as the probability that the effectively received SNR $\gamma \left( t \right)$ at MR is larger than a given SNR threshold $\gamma_{th}$, which can be given by
\begin{equation}\label{eq:pcov}
  P_{\rm{cov}}\left( t \right)=\mathrm{Pr}\left( \gamma \left( t \right) \geqslant \gamma _{\rm{th}} \right) = 1 - \mathrm{Pr}\left( \gamma \left( t \right) < \gamma _{\rm{th}} \right).
\end{equation}

Let $P_{\rm{out}}\left( t \right)$ denote $\mathrm{Pr}\left( \gamma \left( t \right) < \gamma _{\rm{th}} \right)$, and it is given by 
\begin{equation}
\begin{array}{l}
  P_{\rm{out}}\left( t \right)\\
   = \mathrm{Pr}\left( \frac{P\left|h_{\rm{d}} \left( t \right) + \sum_{n=1}^N{h^{r}_{{\rm{n}} }\left( t \right) e^{j\theta _n\left( t \right)}g_{n}}\left( t \right)  \right|^2 }{\sigma^2}  < \gamma _{\rm{th}}    \right) \\
   =\mathrm{Pr} \left( \left| h_{\rm{d}} \left( t \right) + \sum_{n=1}^N{h^{n}_{{\rm{r}} }\left( t \right) e^{j\theta _n\left( t \right)}h^{n}_{{\rm{g}} }}\left( t \right) \right|^2<\frac{\gamma _{\rm{th}}}{\bar{\gamma}}  \right)\\
   =\mathrm{Pr}\left( \left| h\left( t \right) \right|^2<\frac{\gamma _{\rm{th}}}{\bar{\gamma}} \right).
\end{array}
\end{equation}
where $\bar{\gamma}  = \frac{P}{\sigma^2}$ denotes the average transmission SNR.

The type of probability distribution of  $h\left( t \right)$ is given by Theorem~\ref{theorem1}.
\begin{theorem}
  \label{theorem1}
  The equivalent channel  $h\left(t\right)$ between BS and MR, follows complex-valued Gaussian distribution with mean $\mu_h\left(t\right)$ and  variance $\sigma_{h}^2\left(t\right)$, namely, $h\left(t\right) \sim \mathcal{C} \mathcal{N} \left( \mu_h\left(t\right),\sigma_{h}^2\left(t\right) \right)$, and we have 
\begin{align}  
  \mu_h\left(t\right) &=\rho _{\rm{d} }\sqrt{D_{\rm{d}}^{-\varepsilon_{\rm{d}}}\left( t \right)}e^{-j\theta _{ \rm{d}}\left(t\right)} \\ \nonumber
  &+\sum_{n=1}^N\rho_{\rm{r}}\rho _{\rm{g} }\sqrt{\left(d_{\rm{r}}^{n} \left( t \right)\right )^{-\varepsilon_{\rm{r}}}}\sqrt{\left(d_{\rm{g}}^{n}\right )^{-\varepsilon_{\rm{g}}}} \\ \nonumber
  &\times e^{j\left( \theta _n\left(t\right)-\theta_{\rm{r}}^{n}\left(t\right)-\theta_{\rm{g}}^{n} \right)},\nonumber
\end{align}
\begin{align}
\sigma _{h}^{2}\left(t\right)&=\varrho_{\rm{d}}^2d_{\rm{NLoS,d}}^{-\varepsilon'_{\rm{d}}}\left( t \right) \\ \nonumber
&+\sum_{n=1}^{N}\varrho_{\rm{r}}^2\varrho_{\rm{g}}^2\left(d_{\rm{NLoS,r}}^{n}\left( t \right)\right)^{-\varepsilon'_{\rm{r}}}\left(d_{\rm{NLoS,g}}^{n}\right)^{-\varepsilon'_{\rm{g}}}, \nonumber
\end{align}
where $\rho_{\rm{d}}=\sqrt{\frac{\kappa _{{\rm{d}}}}{\kappa _{{\rm{d}}}+1}}$, $\varrho _{{\rm{d}}}=\sqrt{\frac{1}{\kappa _{\rm{d}}+1}}$, $\rho _{\rm{g}}=\sqrt{\frac{\kappa _{{\rm{g}}}}{\kappa _{{\rm{g}}}+1}}$, $\varrho _{{\rm{g}}}=\sqrt{\frac{1}{\kappa _{{\rm{g}}}+1}}$, $\rho _{\rm{r}}=\sqrt{\frac{\kappa _{\rm{r}}}{\kappa _{\rm{r}}+1}}$, and $\varrho _{{\rm{r}}}=\sqrt{\frac{1}{\kappa _{\rm{r}}+1}}$.
\end{theorem}

\begin{proof}
See Appendix A.
\end{proof}

According to Theorem~\ref{theorem1}, the coverage probability can be derived as in Theorem~\ref{theorem2}.
\begin{theorem}
  \label{theorem2}
  $P_{\rm{out}}\left( t \right)$ follows a non-central chi-square distribution, i.e., $\chi ^2 \left(\nu,\zeta \left( t \right)\right)  $, with the degree freedom $\nu = 1$ , and the non-centrality parameter $\zeta\left( t \right) = \frac{\mu_h^2\left(t\right)}{\sigma_{h}^2\left(t\right)}$, where
\begin{equation}
  \begin{array}{l}\left| \mu_h\left( t \right) \right|^2= \left|\rho _{\rm{d} }\sqrt{D_{\rm{d}}^{-\varepsilon_{\rm{d}}}\left( t \right)}e^{-j\theta _{ \rm{d}}\left(t\right)} \right. \\ \left.
  +\sum_{n=1}^N\rho_{\rm{r}}\rho _{\rm{g} }\sqrt{\left(d_{\rm{r}}^{n} \left( t \right)\right )^{-\varepsilon_{\rm{r}}}}\sqrt{\left(d_{\rm{g}}^{n}\right )^{-\varepsilon_{\rm{g}}}}e^{j\left( \theta _n\left( {t} \right) -\theta _{\rm{r}}^{n}\left( {t} \right) -\theta _{\rm{g}}^{n} \right)} \right|^2.
\end{array}
\end{equation}

With the corresponding cumulative distribution function (CDF), $P_{\rm{out}}\left( t \right)$ is given by
\begin{equation}
  P_{\rm{out}}\left( t \right) = 1 - Q_{\frac{1}{2}}\left(\sqrt{\zeta\left( t \right)}, \sqrt{\gamma_0\left( t \right)}\right),
\end{equation}
where $\gamma_0 = \frac{\gamma_{th}}{\bar{\gamma}\sigma _{h}^{2}\left( t \right)}$, and $Q_m \left(a,b \right)$ is the Marcum Q-function defined in \cite{t25}. Thus, the coverage probability can be expressed as 
\begin{align}
  P_{\rm{cov}} = 1 - P_{\rm{out}} = Q_{\frac{1}{2}}\left(\sqrt{\zeta\left( t \right)}, \sqrt{\gamma_0\left( t \right)}\right).
\end{align}
\end{theorem}

\begin{proof}
  See Appendix B.
\end{proof}

\subsection{RIS Phase Optimization}
The set of discrete phase $\varTheta $ contains a series of discrete variables, and the range available for each phase depends on RIS quantization bits. Considering the complexity and validity, we exploit the local search method as shown in Algorithm~\ref{phase} to optimize the phase. Specifically, keeping the other $N-1$ phase values fixed, for each element $\theta_n\left( t \right)$, we traverse all possible values and choose the optimal one. Then, use this optimal solution $\theta _n^\ast\left( t \right)$ as the new value of $\theta _n\left( t \right)$  for the optimization of another phase, until all phases in the set $\varTheta $ are fully optimized. 
\begin{algorithm}[!t]
  \caption{Local Search for Discrete phase}
  \label{phase}
  \begin{algorithmic}[1]
  \REQUIRE
  the toal time $T$ slots, the number of quantization bits $b$
  \ENSURE
    {$\theta _n^\ast\left( t \right),~\forall n \in \mathcal{N}$ }
  \FOR {$t=1:T$}
  \FOR {$n=1:N$}
  \STATE Assign all possible values to $\theta _n\left( t \right)$, and select the value maximizing the coverage probability $P_{\rm{cov}} \left( t \right)$ denoted as $\theta _n^\ast\left( t \right)$;
  \STATE $\theta _n\left( t \right)=\theta _n^\ast\left( t \right)$;
  \ENDFOR
  \ENDFOR
  \end{algorithmic} 
\end{algorithm}

\section{Numerical Analysis}
In this section, we analyze coverage probability of RIS-assisted HST communications. Simulation results are provided to validate system performance in terms of coverage probability. For comparison, the scheme without using RIS is considered, which does not use RIS for signal reflection and MR can only receive signals through BS-MR channels. The simulation parameters are set as listed in Table~\ref{tab:1}. 
\begin{table}[!htbp]   
  \caption{Simulation Parameters}
  \begin{center}\label{tab:1}   
  \begin{tabular}{|c|c|c|}  \hline
    \textbf{Parameter}  & \textbf{Symbol}    & \textbf{Value} \\
    \hline
     {Speed of HST}        & $v$                & {$360$~km/h }    \\  
    \hline
    \ {Height of BS}        & $H_{\rm{BS}}$      &  $10$~m \\
    \hline
     {Height of RIS}       & $H_{\rm{RIS}}$     &  $2$~m \\
    \hline
     {Height of MR}        & $H_{\rm{MR}}$      &  $2.5$~m \\
    \hline
    Bandwidth    & $B$                & $20$~MHz \\
    \hline
    Noise power        & $\sigma^2$        & $-174$~dBm/Hz $+10\log _{10}B+10$~dB \\
    \hline
    Carrier frequency  & $f$               & $2.35$~GHz \\
    \hline
    Rician factor      & $\kappa _{\left( {\rm{d,r,g}}\right)}$ & $10$~dB \\
    \hline
  \end{tabular}
\end{center}
\end{table}

Fi{}g.~{\ref{fig1}} illustrates coverage probability $P_{\rm{cov}}$ against transmission power under two schemes and varying numbers of RIS elements. It can be observed that coverage probabilities under these two schemes increase with increasing of the transmission power. Coverage probability of using RIS outperforms the case without RIS. It can be further observed that as the number of RIS elements $N$ increases, $P_{\rm{cov}}$ increases. This is because more RIS elements in the system result that more signal paths and energy can be reflected to enhance the signal quality at MR.
\begin{figure}[!t]
\centering
\includegraphics[scale=0.45]{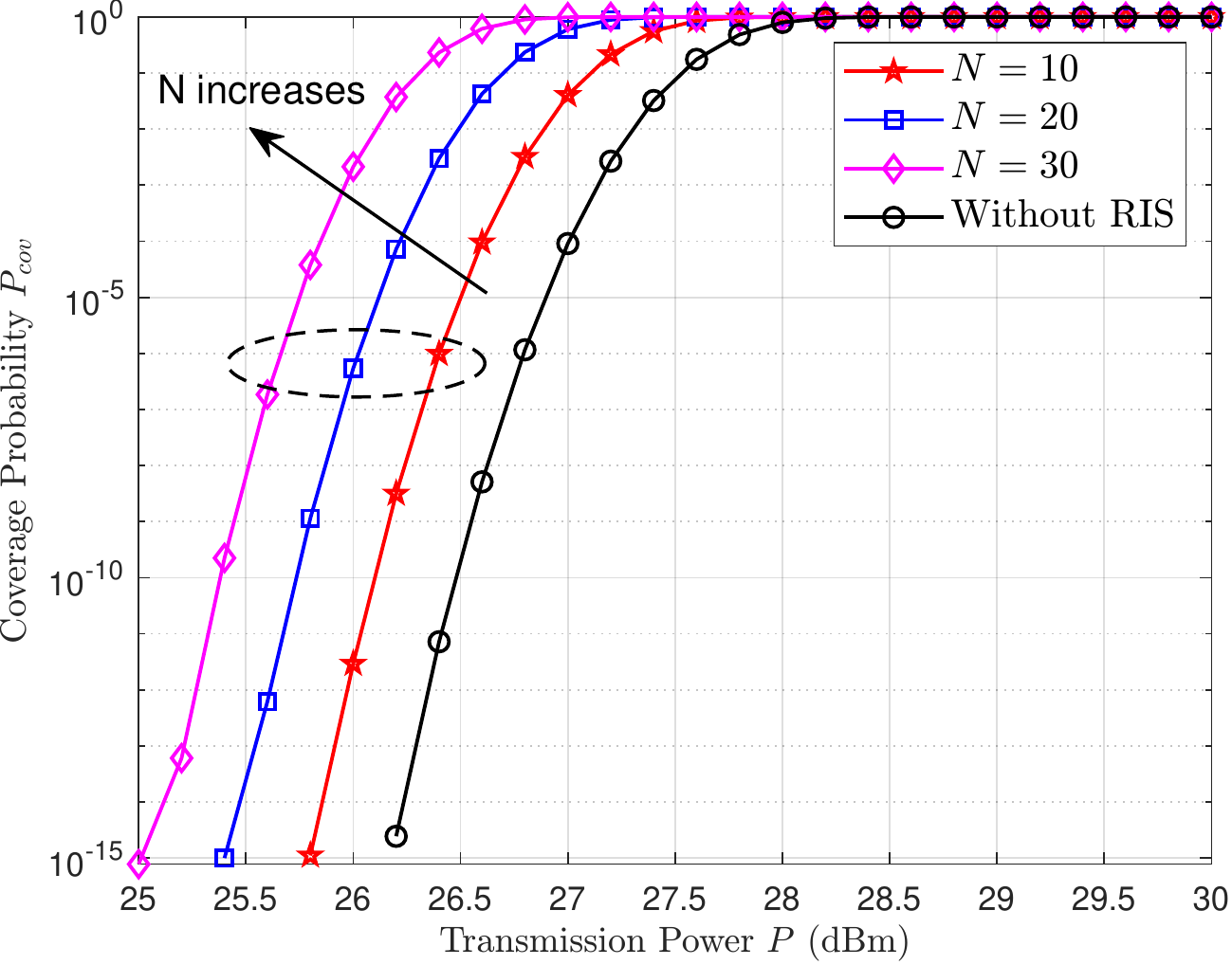}\\
\caption{\label{fig1} Coverage probability vs. transmission power.}
\end{figure}
\begin{figure}[!t]
  \centering
  \includegraphics[scale=0.45]{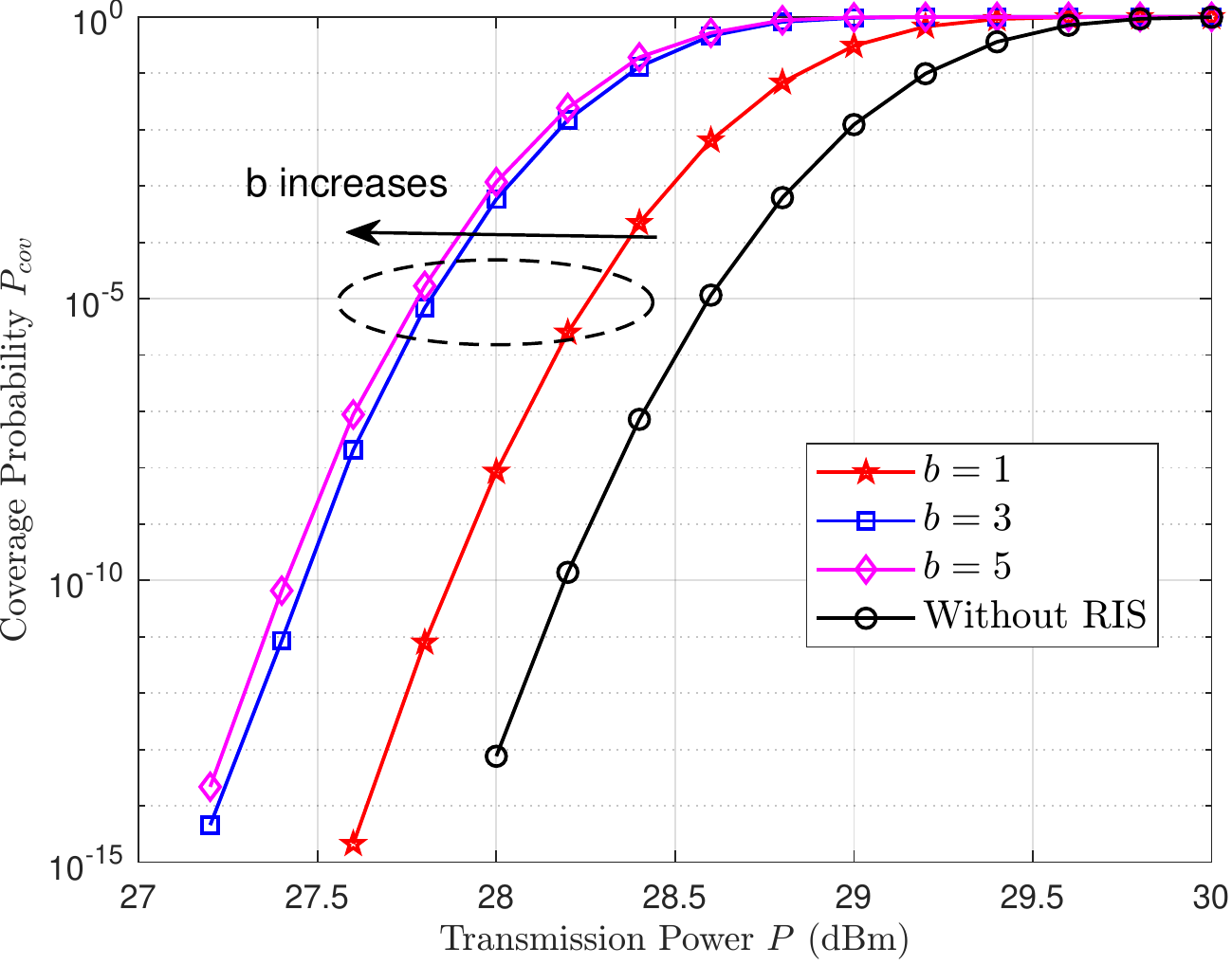}\\
  \caption{\label{fig2} Coverage probability vs. transmission power.}
\end{figure}

Fi{}g.~{\ref{fig2}} shows coverage probability $P_{\rm{cov}}$ against transmission power under different schemes and varying numbers of RIS quantization bits. It can be observed that $P_{\rm{cov}}$ increases with increasing number of RIS quantization bits $b$. It can be observed that coverage probability of $b=1$ is significantly lower than $b=3$ and $b=5$, which means that coverage probability is significantly improved when $b>1$. This is because phase resolution increases with increasing $b$, which results that  the received power at MR is enhanced. It can be further observed that the narrow coverage probability gap exists between $b=3$ and $b=5$, which indicates that coverage probability may remain unchanged with further increasing $b$.

Fi{}g.~{\ref{fig4}} illustrates coverage probability $P_{\rm{cov}}$ versus the SNR threshold $\gamma_{th}$ under different numbers of RIS elements. It can be observed that coverage probabilities under these two schemes decrease with increasing  transmission power. Obviously, improving SNR threshold requirement means more outage events and lower coverage probability. When $N$ increases, channel gain of $h\left( t\right)$ increases, which result that the received power at MR is enhanced and the coverage probability is thus improved.
\begin{figure}[!t]
    \centering
    \includegraphics[scale=0.45]{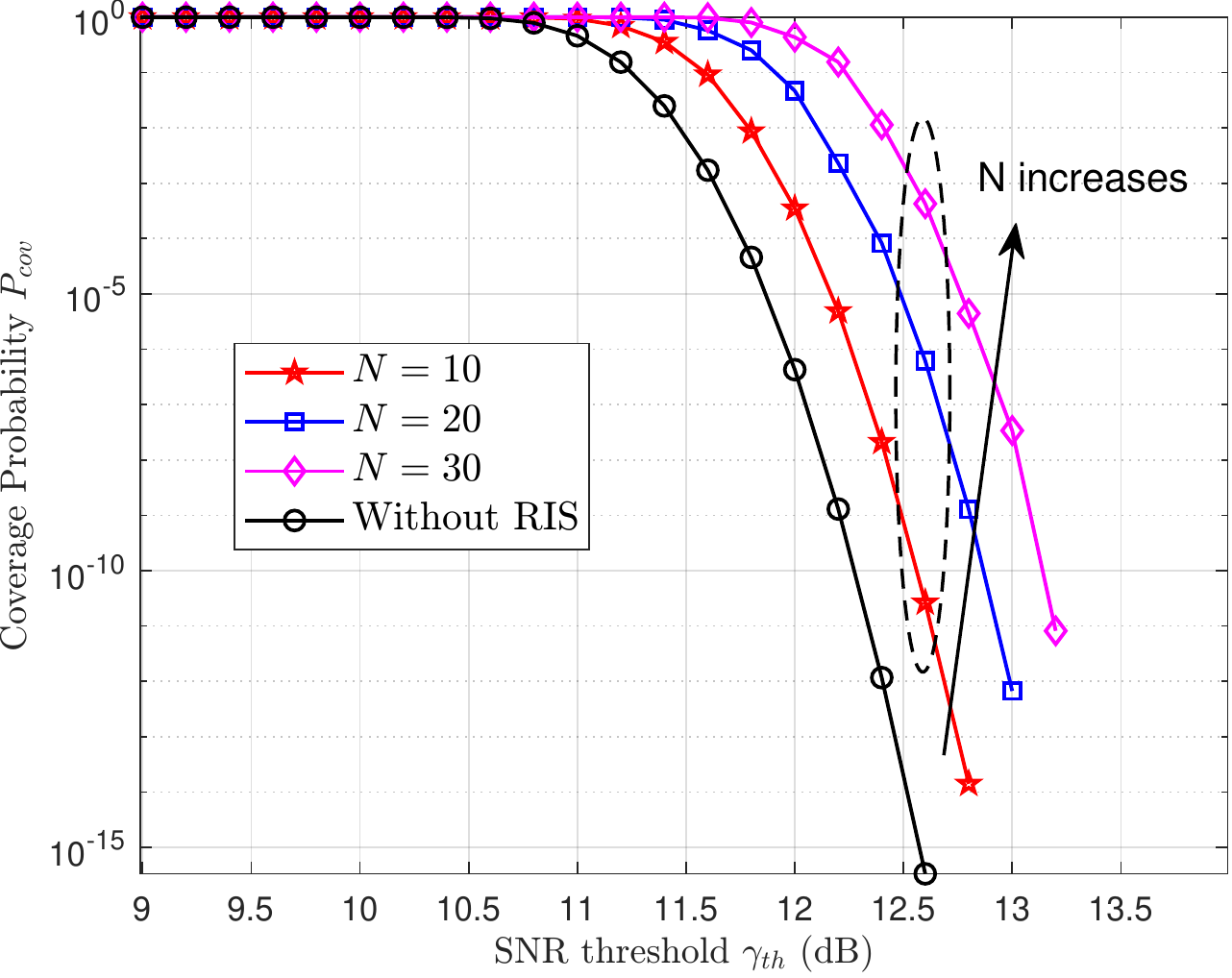}\\
    \caption{\label{fig4} Coverage probability vs. SNR threshold.}
\end{figure}
\begin{figure}[!t]
  \centering
  \includegraphics[scale=0.451]{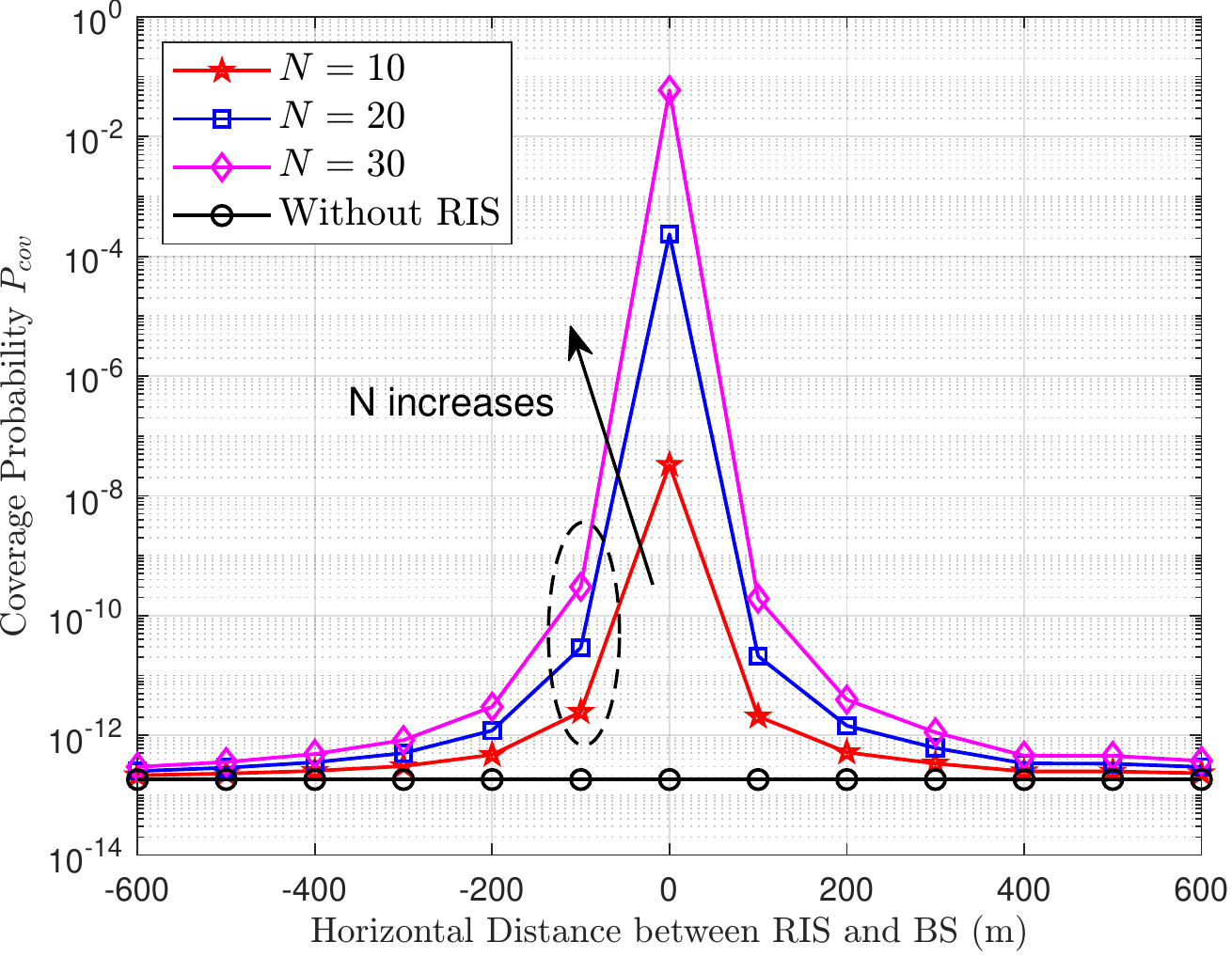}\\
  \caption{\label{fig5}Coverage probability vs. horizontal distance between BS and RIS.}
\end{figure}

Fi{}g.~{\ref{fig5}} shows impact of horizontal distances $d_{\rm{BR}}^{\rm{h}}$ between RIS and BS, as shown in  F{}ig.~\ref{fig:top}, on coverage probability. It is observed that the coverage probability of the cases using RIS firstly increases and then decreases, while the case with using RIS dose not change with location. It is found that when horizontal distance $d_{\rm{BR}}^{\rm{h}}$ is equal to $-600$ m or $600$ m, i.e., at the edge of the investigated area, the coverage probability is lowest. While horizontal distance $d_{\rm{BR}}^{\rm{h}}$ is equal to $0$, which is the center of the long axis of BS elliptical coverage area, the coverage probability is high. This is owing to the fact that when RIS moves towards the center of the area, path loss of reflection link changes accordingly, which leads to an enhancement of the reflected signal and the benefits of RIS are fully utilized. This show that system performance is sensitive to placement of RIS.

\section{Conclusion}
In this paper, we investigate wireless coverage probability analysis of downlink SISO RIS-assisted HST communication system. An closed-form expression of the coverage probability is derived. We have analyzed the impact of different system parameters on coverage probability including transmission power, SNR threshold, and horizontal distance between BS and RIS. Numerical results have demonstrated that better coverage performance can be achieved by using well RIS. The results in this paper can serve as a guidance for RIS-assisted HST communication coverage analysis.

\section*{Appendix A}
Substituting \eqref{eq:hbm}, \eqref{eq:hbr} and \eqref{eq:hrm} into $h\left(t\right)$, we have
\begin{align}\label{eq:h2}
    &  h\left( t \right) = \rho _{\rm{d}}\bar{h}_{ \rm{d}}\left( t \right) + \varrho _{{\rm{d}}}\tilde{h}_{ \rm{d}}\left( t \right) + \sum_{n=1}^N{\rho _{\rm{r}}\rho _{\rm{g}}\bar{h}^{n}_{{\rm{r}} }\left( t \right) e^{j\theta _n\left( t \right)}\bar{h}^{n}_{{\rm{g}} }} \\ \nonumber
       &+ \sum_{n=1}^N{\rho _{\rm{r}}\varrho _{{\rm{g}}}\bar{h}^{n}_{{\rm{r}} }\left( t \right) e^{j\theta _n\left( t \right)}\tilde{h}^{n}_{{\rm{g}} }}\left( t \right) + \sum_{n=1}^N{\varrho _{\rm{r}}\rho _{{\rm{g}}}\tilde{h}^{n}_{{\rm{r}} }\left( t \right) e^{j\theta _n\left( t \right)}\bar{h}^{n}_{{\rm{g}} }} \\ \nonumber
      &+ \sum_{n=1}^N{\varrho _{\rm{r}}\varrho _{{\rm{g}}}\tilde{h}^{n}_{{\rm{r}} }\left( t \right) e^{j\theta _n\left( t \right)}\tilde{h}^{n}_{{\rm{g}} }}\left( t \right),
      \nonumber 
\end{align}
where $\rho_{\rm{d}}=\sqrt{\frac{\kappa _{{\rm{d}}}}{\kappa _{{\rm{d}}}+1}}$, $\varrho _{{\rm{d}}}=\sqrt{\frac{1}{\kappa _{\rm{d}}+1}}$, $\rho _{\rm{g}}=\sqrt{\frac{\kappa _{{\rm{g}}}}{\kappa _{{\rm{g}}}+1}}$, $\varrho _{{\rm{g}}}=\sqrt{\frac{1}{\kappa _{{\rm{g}}}+1}}$, $\rho _{\rm{r}}=\sqrt{\frac{\kappa _{\rm{r}}}{\kappa _{\rm{r}}+1}}$, and $\varrho _{{\rm{r}}}=\sqrt{\frac{1}{\kappa _{\rm{r}}+1}}$.
Note that, the LoS components of  BS-MR link,  RIS-MR link and  BS-RIS link depend on the corresponding link distances. For a given location, the components $\rho _{\rm{d}}\bar{h}_{ \rm{d}}\left( t \right)$ and  $\sum_{n=1}^N{\rho _{\rm{r}}\rho _{\rm{g}}\bar{h}^{n}_{{\rm{r}} }\left( t \right) e^{j\theta _n\left( t \right)}\bar{h}^{n}_{{\rm{g}}}}$ of \eqref{eq:h2} turn to be deterministic. Since the NLoS component $\tilde{h}_{ \rm{d}}\left( t \right)$ follows complex Gaussian distribution with zero mean and variance $d_{\rm{NLoS,d}}^{-\varepsilon'_{\rm{d}}}\left( t \right)$, $\tilde{h}_{ \rm{r}}\left( t \right)$ and $\tilde{h}_{ \rm{g}}\left( t \right)$ follow complex Gaussian distribution with zero mean and variance $\left(d_{\rm{NLoS,r}}^{n}\left( t \right)\right) ^{-\varepsilon'_{\rm{r}}}$, $\left(d_{\rm{NLoS,g}}^{n}\right) ^{-\varepsilon'_{\rm{g}}}$, respectively, and the parts $\varrho _{{\rm{d}}}\tilde{h}_{ \rm{d}}\left( t \right)$, $\sum_{n=1}^N{\rho _{\rm{r}}\varrho _{{\rm{g}}}\bar{h}^{n}_{{\rm{r}} }\left( t \right) e^{j\theta _n\left( t \right)}\tilde{h}^{n}_{{\rm{g}} }}\left( t \right)$, $\sum_{n=1}^N{\varrho _{\rm{r}}\rho _{{\rm{g}}}\tilde{h}^{n}_{{\rm{r}} }\left( t \right) e^{j\theta _n\left( t \right)}\bar{h}^{n}_{{\rm{g}}}}$ and $\sum_{n=1}^N{\varrho _{\rm{r}}\varrho _{{\rm{g}}}\tilde{h}^{n}_{{\rm{r}} }\left( t \right) e^{j\theta _n\left( t \right)}\tilde{h}^{n}_{{\rm{g}} }}\left( t \right)$ of \eqref{eq:h2} also follow a Gaussian distribution.

The expectation of $h\left( t \right)$ can be written by 
\begin{align} \label{eq:eh2}
  &\mu_h \left(t\right) \triangleq \mathbb{E}\left\{ h\left( t \right)\right\} \\ \nonumber
&= \rho _{\rm{d}}\bar{h}_{ \rm{d}}\left( t \right) +\sum_{n=1}^N{\rho _{\rm{r}}\rho _{\rm{g}}\bar{h}^{n}_{{\rm{r}} }\left( t \right) e^{j\theta _n\left( t \right)}\bar{h}^{n}_{{\rm{g}} }} \\  \nonumber
&=\rho _{\rm{d} }\sqrt{D_{\rm{d}}^{-\varepsilon_{\rm{d}}}\left( t \right)}e^{-j\theta ^{ \rm{d}}\left(t\right)} \\ \nonumber
&+\sum_{n=1}^N\rho_{\rm{r} }\rho _{\rm{g} }\sqrt{\left(d_{\rm{r}}^{n} \left( t \right)\right )^{-\varepsilon_{\rm{r}}}}\sqrt{\left(d_{\rm{g}}^{n}\right )^{-\varepsilon_{\rm{g}}}}e^{j\left( \theta _n\left(t\right)-\theta_{\rm{r}}^{n}\left(t\right)-\theta_{\rm{g}}^{n} \right)} \nonumber.
\end{align}

The variance of $h\left( t \right)$  is derived as 
\begin{align} \label{eq:hsigama}
  &\sigma _{h}^{2}\left(t\right) \triangleq \rm{var}\left\{ h\left( t \right)\right\}\\ \nonumber
  &= \varrho _{{\rm{d}}}^2\rm{var}\left\{ \tilde{h}_{ \rm{d}}\left( t \right) \right\} +\varrho _{\rm{r}}^2\varrho _{{\rm{g}}}^2\rm{var}\left\{ \sum_{n=1}^N{\tilde{h}^{n}_{{\rm{r}} }\left( t \right) e^{j\theta _n\left( t \right)}\tilde{h}^{n}_{{\rm{g}} }}\left( t \right) \right\} \\  \nonumber
  &=\varrho_{\rm{d}}^2d_{\rm{NLoS,d}}^{-\varepsilon'_{\rm{d}}}\left( t \right) +\sum_{n=1}^{N}\varrho_{\rm{r}}^2\varrho_{\rm{g}}^2\left(d_{\rm{NLoS,r}}^{n}\left( t \right)\right)^{-\varepsilon'_{\rm{r}}}\left(d_{\rm{NLoS,g}}^{n}\right)^{-\varepsilon'_{\rm{g}}}. \nonumber
\end{align}

Therefore, $h\left( t \right)$ is proved to follow a complex-valued Gaussian distribution, $h\left(t\right) \sim \mathcal{C} \mathcal{N} \left( \mu_h\left(t\right),\sigma_{h}^2\left(t\right) \right)$. This completes the proof.

\section*{Appendix B}
Due to the Gaussian channel proved above, $h\left(t\right) \sim \mathcal{C} \mathcal{N} \left( \mu_h\left(t\right),\sigma_{h}^2\left(t\right) \right) $. Therefore, $\frac{  \left| h\left(t\right) \right|^2 }{\sigma_{h}^2\left(t\right)}$ follows the non-central chi-squared distribution, i.e.,  $\chi ^2(\nu,\zeta \left( t \right)) $, with the degrees of freedom $\nu = 1$, and the non-centrality parameter is in \eqref{eq:notcenter}.
\begin{figure*}[!t]
\begin{align}
 \label{eq:notcenter}
   \zeta \left( t \right) = \frac{ \left| \mu_h\left(t\right)\right|^2 }{\sigma_{h}^2\left(t\right)} 
  = \frac{\left|\rho _{\rm{d} }\sqrt{D_{\rm{d}}^{-\varepsilon_{\rm{d}}}\left( t \right)}e^{-j\theta ^{ \rm{d}}\left(t\right)}+\sum_{n=1}^N\rho_{\rm{r} }\rho _{\rm{g} }\sqrt{\left(d_{\rm{r}}^{n} \left( t \right)\right )^{-\varepsilon_{\rm{r}}}}\sqrt{\left(d_{\rm{g}}^{n}\right )^{-\varepsilon_{\rm{g}}}}e^{j\left( \theta _n\left(t\right)-\theta_{\rm{r}}^{n}\left(t\right)-\theta_{\rm{g}}^{n} \right)}  \right|^2}{ \varrho_{\rm{d}}^2d_{\rm{NLoS,d}}^{-\varepsilon'_{\rm{d}}}\left( t \right) +\sum_{n=1}^{N}\varrho_{\rm{r}}^2\varrho_{\rm{g}}^2\left(d_{\rm{NLoS,r}}^{n}\left( t \right)\right)^{-\varepsilon'_{\rm{r}}}\left(d_{\rm{NLoS,g}}^{n}\right)^{-\varepsilon'_{\rm{g}}}}.
\end{align}
\hrulefill
\end{figure*}

With the corresponding CDF of $\chi ^2(\nu,\zeta \left( t \right)) $, $P_{\rm{out}}\left( t \right)$ defined in \eqref{eq:pcov} is given by
\begin{equation}
  P_{\rm{out}}\left( t \right) = 1 - Q_{\frac{1}{2}}\left(\sqrt{\zeta\left( t \right)}, \sqrt{\gamma_0\left( t \right)}\right),
\end{equation}
where $\gamma_0 = \frac{\gamma_{th}}{\bar{\gamma}\sigma _{h}^{2}\left( t \right)}$, and $Q_m \left(a,b \right)$ is the Marcum Q-function defined in \cite{t25}. As a result, the coverage probability defined in \eqref{eq:pcov} can be rewrite as 
\begin{align}
  P_{\rm{cov}} = Q_{\frac{1}{2}}\left(\sqrt{\zeta\left( t \right)}, \sqrt{\gamma_0\left( t \right)}\right).
\end{align}
This completes the proof.

\end{document}